\documentclass{amsart}
\usepackage{verbatim}
\usepackage{enumerate}
\usepackage{hyperref}
\usepackage{mathtools}
\usepackage{epigraph}

\newcommand{\fde}[1]{\frac\de{\de #1}}

\newcommand{\lde}[1]{\overrightarrow{\de{#1}}}
\newcommand{\rde}[1]{\overleftarrow{\de{#1}}}

\newcommand{\fullref}[1]{\ref{#1} on page~\pageref{#1}}

\newcommand{\ndash}{\nobreakdash-\hspace{0pt}}
\newcommand{\Ndash}{\nobreakdash--}

\newcommand{\dd}{{\mathrm{d}}}

\newcommand{\EE}{\mathrm{e}}

\DeclareMathOperator{\Dens}{Dens}

\DeclareMathOperator{\Ber}{Ber}

\newcommand{\Div}{{\mathrm{div}}}

\DeclareMathOperator{\End}{End}

\newtheorem{Thm}{Theorem}[section]
\newtheorem{Prop}[Thm]{Proposition}
\newtheorem{Lem}[Thm]{Lemma}
\newtheorem{Cor}[Thm]{Corollary}

\newtheorem*{Thm*}{Theorem}
\newtheorem*{Lem*}{Lemma}

\theoremstyle{remark}
\newtheorem*{Ack}{Acknowledgment}

\theoremstyle{definition}
\newtheorem{Rem}[Thm]{Remark}%[section]

\newtheorem{Def}[Thm]{Definition}%[section]
\newtheorem{ThmDef}[Thm]{Theorem/Definition}
\newcommand{\LL}{\mathsf{L}}

\newcommand{\bbR}{{\mathbb{R}}}

\newcommand{\bbZ}{{\mathbb{Z}}}

\newcommand{\de}{\partial}

\newcommand{\calD}{\mathcal{D}}

\newcommand{\calL}{\mathcal{L}}
\newcommand{\calM}{\mathcal{M}}
\newcommand{\calN}{\mathcal{N}}

\def\gpd{\,\lower1pt\hbox{$\longrightarrow$}\hskip-.24in\raise2pt
               \hbox{$\longrightarrow$}\,}

\let\Tilde=\widetilde

\let\Hat=\widehat

\newcommand\qq{}
\newcommand\cmp[1]{{\qq Commun.\ Math.\ Phys.\ \bf #1}}
\newcommand\jmp[1]{{\qq J.\ Math.\ Phys.\ \bf #1}}
\newcommand\pl[1]{{\qq Phys.\ Lett.\ \bf #1}}

\newcommand\mpl[1]{{\qq Mod.\ Phys.\ Lett.\ \bf #1}}

\newcommand\lmp[1]{{\qq Lett.\ Math.\ Phys.\ \bf #1}}

\begin{document}
\title{The Canonical BV Laplacian on Half\ndash Densities}

\begin{abstract}
%This space intentionally left blank.
This is a didactical review on the canonical BV Laplacian on half\ndash densities.
\end{abstract}

\author{Alberto S. Cattaneo}
\address{Institut f\"ur Mathematik, Universit\"at Z\"urich\\
Winterthurerstrasse 190, CH-8057 Z\"urich, Switzerland}  
\email{cattaneo@math.uzh.ch}

\epigraph{\emph{Das Bekannte \"uberhaupt ist darum, weil es bekannt ist, nicht erkannt}}{GWFH}

\thanks{ I acknowledge partial support of the SNF Grant No.\ 200020\_192080, of the Simons Collaboration on Global Categorical Symmetries, and of the COST Action 21109 %Cartan geometry, Lie, Integrable Systems, quantum group Theories for Applications 
(CaLISTA). This research was (partly) supported by the NCCR SwissMAP, funded by the Swiss National Science Foundation.}

%\keywords{%Hamilton--Jacobi, Batalin--Vilkovisky, Segal--Bargmann, Chern--Simons, Wess--Zumino--Witten--Novikov,
%%Calabi--Yau, Kodaira--Spencer
%Hamilton--Jacobi, (generalized) generating functions,
%%Chern--Simons theory,
%Chern--Simons, Wess--Zumino--Witten,
%%, Batalin--Vilkovisky formalism, Segal--Bargmann transform, Wess--Zumino--Witten model
% nonlinear (Hitchin) phase space polarization,  Kodaira--Spencer (BCOV) action,
% Batalin--Vilkovisky, Batalin--Fradkin--Vilkovisky}

%\subjclass[2020]{81T70, % "Quantization in field theory; cohomological methods" (MSC2020)
%53D22, % "Canonical transformations in symplectic and contact geometry" (MSC2020)
%70H20, % "HamiltonÐJacobi equations in mechanics" (MSC2020)
%53D55, % "Deformation quantization, star products" (MSC2020)
%53D50 (Primary); % "Geometric quantization" (MSC2020)
%81T13, %: "YangÐMills and other gauge theories in quantum field theory" (MSC2020)
%81S10, %: "Geometry and quantization, symplectic methods" (MSC2020)
%70H15, %: "Canonical and symplectic transformations for problems in Hamiltonian and Lagrangian mechanics" (MSC2020)
%57R56, %: "Topological quantum field theories (aspects of differential topology)" (MSC2020)
%81T45 (Secondary). % "Topological field theories in quantum mechanics" (MSC2020)
%}

%\dedicatory{In memory of Kirill Mackenzie}

\maketitle

\tableofcontents

%\allowdisplaybreaks

\section{Introduction}
The BV Laplacian on functions is a central ingredient of the Batalin--Vilkovisky (BV) formalism. Its global definition is not canonical, yet
a remarkable fact observed by Khudaverdian \cite{Khu99,Khu04} is that
its extension to half\ndash densities is so, %canonical 
and there are various ways to show it. %cite K, KV, S
In this review we give a short, self-contained presentation which requires just a few simple computations. It gives a special emphasis on the fact that half-densities have a natural inner product.

This note has no pretence of originality. In Appendix~\ref{a:historem}, we give a short historical overview with the relevant references.

In Appendix~\ref{a:somebackgroundmaterial}, we collect some technical background material.
In Appendix~\ref{a:app}, we recall why the BV formalism is important in applications.

%The various results are scattered in the literature. 
%We just try to provide a coherent, ready-to-use presentation.

\subsection{Overview}
We recall definitions and properties of BV Laplacians  (on functions and densities) on an oriented odd symplectic manifold with the goal of
proving that on half\ndash densities it is canonically defined (whereas on other kinds of densities, including functions, it requires the choice of a compatible density). This goes as follows:
\begin{enumerate}
\item On functions we define $\Delta_\mu f = \frac12 \Div_\mu X_f$, where $\mu$ is an even, nowhere vanishing\footnote{Throughout the paper, nowhere vanishing means nonvanishing at every point of the body after the nilpotents are set to zero. %An even, nowhere vanishing density is the same as a basis for the module of densities. 
An even, nowhere vanishing density $\mu$ is the same as a basis for the module of densities. The $k$\ndash density $\mu^k$ is then a basis for the module of $k$\ndash densities.} density, and $X_f$ denotes the hamiltonian vector field of $f$.
\item On half-densities we define $\Delta^{(\frac12)}_\mu \sigma = \Delta_\mu(\sigma / \mu^{1/2}) \mu^{1/2}$.
\item We check that the leading term %of each of these differential operators 
is invariant under symplectomorphisms and
does not depend on the choice of $\mu$.
\item\label{over:symm} We check that $\Delta^{(\frac12)}_\mu$ is a symmetric operator. This implies that in any Darboux chart 
only the zeroth-order term %of the operator 
may depend on $\mu$.
\item We show that for a compatible density $\mu$, i.e., for $\mu$ satisfying $\Delta_\mu^2=0$, the zeroth-order term of the previous point must vanish. This shows that $\Delta^{(\frac12)}_\mu$ does not depend on the choice of the compatible density $\mu$.
\item We show, via the presentation of the given odd symplectic manifold as a symplectomorphic odd cotangent bundle, that compatible densities always exist.
\item We conclude that the BV Laplacian on half-densities, with the property of squaring to zero, is canonically defined.
\item As an aside, we also show how to define this canonical operator, denoted by $\Delta$, directly on Darboux charts without any choice involved. Moreover, we prove that the compatibility condition for a density $\mu$ can also be expressed as $\Delta\mu^{\frac12}=0$
and that the BV Laplacian on functions can now be recovered as
$\Delta_\mu = \Delta(f\mu^{\frac12})/\mu^{\frac12}$.
\end{enumerate}

We go through details, also checking signs carefully, to make this presentation useful also as a reference. However, it should be observed that, roughly, the above flow can be easily  checked without going into sign details---the only crucial point being (\ref{over:symm}), which requires some care.

\begin{Ack}
I thank A.~Cabrera, D.~Fiorenza, N.~Moshayedi, M.~Schiavina, and P.~\v Severa for useful remarks. I am especially grateful to T.~Voronov for several extremely helpful exchanges.
\end{Ack}

\section{The standard BV Laplacian}\label{s-standLap}
On functions of odd variables $p_1,\dots,p_n$ and even variables $q^1,\dots,q^n$, the standard BV Laplacian\footnote{This operator was introduced to show that certain integrals are invariant under deformations of the integration domain; see Appendix~\ref{a:app}.} (a.k.a.\ BV Laplace operator or BV operator) is defined as
\[
\triangle\coloneqq \de_i\de^i,
\]
where we use Einstein's summation convention and the shorthand notations
\[
\de_i = \fde{q^i}\quad\text{and}\quad\de^i=\fde{p_i}.
\]

Another ingredient of the BV formalism is the BV bracket
%For reference, we spell out our conventions:
\begin{equation}\label{standBVbrack}
(f,g) \coloneqq f\rde{^i}\,\lde{_i^{\phantom{i}}}g-f\rde{_i^{\phantom{i}}}\,\lde{^i}g,
\end{equation}
where $f$ and $g$ are functions of homogeneous degree, and
the arrows denote the directions the derivatives are applied from. Explicitly, $\lde{_i^{\phantom{i}}}f=\de_if$, $\lde{^i}f=-\de^if$, and
\[
\begin{split}
f\rde{_i^{\phantom{i}}} &=\de_i f,\\
f\rde{^i} &= -(-1)^f\de^if.
\end{split}
\]
Here and in the following, when we put a function (or any other object) as an exponent of $(-1)$ we mean its degree.\footnote{\label{f:signs}A more precise but more cumbersome notation would be $(-1)^{|f|}$, where $|f|$ denotes the degree of $f$.}

The standard BV Laplacian has the following properties which are easily proved by direct computation:
\begin{subequations}\label{es:triangle}
\begin{gather}
\triangle^2=0,\label{e:trianglesquare}\\
\triangle(fg)=(\triangle f)g +(-1)^f f\triangle g - (-1)^f (f,g),\label{e:trianglefg} \\
\triangle \EE^S = \left(\triangle S-\frac12(S,S)\right)\EE^S.\label{e:triangleE}
\end{gather}
\end{subequations}
Here again $f$ and $g$ are functions of homogeneous degree, whereas 
the function $S$ is assumed to be even. As an exercise, just using \eqref{e:trianglesquare} and \eqref{e:trianglefg}, you can also show that
\begin{equation}\label{e:triangle(fg)}
\triangle(f,g)=(\triangle f,g) - (-1)^f(f,\triangle g).
\end{equation}

The geometric interpretation of the content of this section is that the $p_i$s and $q^i$s are Darboux coordinates for the odd symplectic form $\omega=\dd p_i\dd q^i$ with $(\ ,\ )$ its associated odd Poisson bracket.\footnote{If the $q^i$s are coordinates on an open subset $U$ of $\bbR^n$, the functions of the $p$ and $q$ variables can be identified with multivector fields on $U$. The BV bracket then gets identified with the Schouten--Nijenhuis bracket. We will return to this, with a general basis manifold, in Section~\ref{s:oct}; see footnote~\fullref{f-mv}.} The main problem with this, however, is that the standard BV Laplacian does not transform well under symplectomorphisms, so it cannot be used as such to define an operator on functions on an odd symplectic manifold. In the next sections, we will see how to obviate this problem.

\section{The BV bracket}
Let $(\calM,\omega)$ be an odd symplectic manifold. Given a function $f$, we denote by $X_f$  its hamiltonian vector field defined via
\[
\iota_{X_f}\omega=\dd f.
\]
Given a second function $g$, we define the BV bracket as
\[
(f,g)\coloneqq(-1)^{f+1}X_f(g).
\]
This is an odd Poisson bracket, which in local Darboux coordinates agrees with the one of the previous section:\footnote{This is a simple computation which however requires checking signs. It is just needed to make sure that the conventions we choose are compatible.}
%\footnote{There is some arbitrariness in the choice of signs, but one has to be careful to choose them compatibly so that in local Darboux coordinates one gets back the BV bracket as defined in the previous section.}
\begin{Lem}\label{l:BVbraDar}
In Darboux coordinates, $\omega=\dd p_i\dd q^i$, the BV bracket agrees with the standard one defined in \eqref{standBVbrack}.
\end{Lem}
\begin{proof}
Writing $X= X^i\de_i+X_i\de^i$, we get\footnote{We use the convention of total degree, i.e., internal degree plus parity of the form degree. This means that $\dd$ is odd and that $\iota_X$ has parity opposite to that of $X$. Moreover, $X^i$ has the same parity as $X$, whereas $X_i$ has opposite parity. Finally, $f$ and $X_f$ have opposite parity to each other.}
\[
\iota_{X}\omega = \dd p_i\,X^i + X_i\,\dd q^i = \dd p_i\,X^i - (-1)^X \dd q^i\,X_i.
\]
{}From
\[
\dd f = \dd p_i\, \de^i f + \dd q^i\, \de_i f
\]
we then get
\begin{subequations}\label{e:Xf}
\begin{align}
(X_f)^i &= \de^i f = (-1)^{f+1} f\rde{^i}
\intertext{and}
(X_f)_i &= (-1)^f\de_i f = (-1)^f f\rde{_i^{\phantom{i}}},
\end{align}
\end{subequations}
so $(-1)^{f+1}X_f(g)$ is given by \eqref{standBVbrack}.
\end{proof}

The BV bracket satisfies several properties.
\begin{Prop} For all functions $f$, $g$, and $h$, we have
\begin{align*}
[X_f,X_g] &= X_{(f,g)},\\
(f,gh)&=(f,g)\,h+(-1)^{(f+1)g} g\,(f,h),\\
(g,f) &= -(-1)^{(f+1)(g+1)}(f,g),\\
(f,(g,h)) &= ((f,g),h)+(-1)^{(f+1)(g+1)}(g,(f,h)).
\end{align*}
\end{Prop}
The last three properties say that $(\ ,\ )$ is an odd Poisson bracket,\footnote{In particular this means that $(\ ,\ )$ is a super Lie bracket with respect to the opposite parity of its arguments.} whereas the first says the the map $f\mapsto X_f$ is a morphism of super Lie algebras.
\begin{proof}
We prove only the first identity. The others are also easily obtained from the definitions (or recovered from the identical identities for the standard BV bracket which one obtains in each Darboux chart). We have
\begin{gather*}
\iota_{[X_f,X_g]}\omega= [\LL_{X_f},\iota_{X_g}]\omega = \LL_{X_f}\iota_{X_g}\omega = \LL_{X_f}\dd g\\
= (-1)^{f+1}\dd\LL_{X_f}g = (-1)^{f+1}\dd X_f(g) = \dd(f,g) = \iota_{X_{(f,g)}}\omega,
\end{gather*}
where $\LL_{X_f}$ denotes the Lie derivative with respect to $X_f$.
\end{proof}

\section{The BV Laplacian on functions}\label{s:BVlapf}
Let $(\calM,\omega)$ be an odd symplectic manifold and
$\mu$  an even, nowhere vanishing density on $\calM$, which we assume to be orientable. We then define the $\mu$\ndash dependent BV Laplacian as
\begin{equation}\label{e:Deltamudiv}
\Delta_\mu f\coloneqq \frac12\Div_\mu X_f.
\end{equation}
Recall (or see Appendix~\ref{a:dens}) that the divergence operator of a vector field $X$ with respect to an even, nowhere vanishing density $\mu$ is defined via 
\[
\Div_\mu X \,\mu= \LL_X\mu,
\]
where $\LL_X$ denotes the Lie derivative.

Again we have to make sure that this definition of the BV Laplacian agrees with the standard one in the appropriate case:
\begin{Prop}\label{p:DeltafD}
In Darboux coordinates, $\omega=\dd p_i\dd q^i$, with standard density $\mu_\text{stand}\coloneqq \dd^nq\,\dd^np$, we have
\[
\Delta_{\mu_\text{stand}} = \triangle.
\]
\end{Prop}
\begin{proof}
With the standard density the divergence of a vector field $X= X^i\de_i+X_i\de^i$ is given by \eqref{se:divstand}, i.e.,
\[
\Div_{\mu_\text{stand}} X = \de_i X^i-(-1)^X\de^iX_i.
\]
If we now insert $X_f$ as calculated in \eqref{e:Xf},
%at the end of the proof of Lemma~\ref{l:BVbraDar},
we get
\[
\Div_{\mu_\text{stand}} X_f = 2\de_i\de^if.
\]
\end{proof}

Next we check how the operator depends on the choice of density. Note that, given two even, nowhere vanishing densities $\mu$ and $\Tilde\mu$, there is a unique even, nowhere vanishing function $h$ such that $\Tilde\mu=h\mu$.
\begin{Prop}\label{p:Deltahmu} We have
\[
\Delta_{h\mu} = \Delta_\mu +\frac12\frac1h X_h.
\]
%Therefore, the leading term of the BV Laplacian on functions is independent of the choice of density.
\end{Prop}
\begin{proof}
The divergence operator depends on the choice of density as
\[
\Div_{h\mu} X_f = \Div_\mu X_f + \frac1hX_f(h)
\]
(see point~(\ref{p:div:hmu}) in Proposition~\ref{p:div}). However, $X_f(h)=(-1)^{f+1}(f,h)=-(h,f)=X_h(f)$,
%$X_f(h)=(f,h)=(h,f)=X_f(h)$, 
since $h$ is even.
\end{proof}
%\begin{Rem}\label{r:DeltaD}
We can always write $h=\pm\EE^S$ with $S$ an even function. This way the formula simplifies to
\[
\Delta_{h\mu} = \Delta_\mu +\frac12 X_S,
\]
since $X_{\EE^S}=\EE^SX_S$.
In particular, we have the
\begin{Lem}\label{l:DeltaD}
In Darboux coordinates with $\mu=\pm\EE^S\mu_\text{stand}$, %thanks to Proposition~\ref{p:DeltafD}, 
we have
\begin{equation}\label{e:DeltaD}
\Delta_\mu=\triangle  +\frac12 X_S=\triangle-\frac12(S,\ ).
\end{equation}
\end{Lem}
This, together with Proposition~\ref{p:Deltahmu}, gives the
\begin{Prop}\label{p:DeltaD}
The BV Laplacian on functions is a second-order differential operator, and its leading term is independent of the choice of density.
\end{Prop}
%In general the property of equation \eqref{e:trianglesquare} does not extend. Let us check this in a Darboux chart using the formula of Remark~\ref{r:DeltaD}:
%\[
%\Delta_\mu f = \triangle f -\frac12(S,f).
%\]
\begin{Rem}\label{r:DeltaD} Under the assumptions of Lemma~\ref{l:DeltaD},
from \eqref{e:DeltaD}, also using \eqref{e:triangle(fg)} and the Jacobi identity, we get
\[
\Delta_\mu^2 f = -\frac12 (F_S,f)
\]
with
\begin{equation}\label{e:FS}
F_S \coloneqq \triangle S-\frac14(S,S).
\end{equation}
We then see that $\Delta_\mu^2=0$ if and only if $F_S$ is constant. As $F_S$ is odd, this happens if and only if
$F_S=0$.\footnote{Note that this condition can also be rephrased as 
$\triangle\EE^{\frac12 S}=0$.
The appearance of the factor $\frac12$ in the exponent will become clear when we study half\ndash densities (see Lemma~\ref{l:hattilde}).} Therefore, in general the property of equation \eqref{e:trianglesquare} does not extend to the global BV Laplacian.
\end{Rem}
%\end{Rem}

\begin{Def}
We say that an even, nowhere vanishing density 
$\mu$ is compatible with the odd symplectic structure if $\Delta_\mu^2=0$. 
\end{Def}
\begin{Rem}
We will see that every odd symplectic manifold admits a compatible density.
\end{Rem}

\begin{Rem}\label{r:tildeDelta}
As we have seen in Remark~\ref{r:DeltaD}, the property of equation \eqref{e:trianglesquare} does not extend in general. One can try to obviate it by modifying the definition of the BV operator. For example, in Darboux coordinates with $\mu=\pm\EE^S\mu_\text{stand}$, we could define
\[
\Tilde\Delta_\mu\coloneqq \Delta_\mu + F_S = \triangle  +\frac12 X_S + \frac12 F_S
\]
with $F_S$ as in \eqref{e:FS}. It is an easy exercise (show first that $\Delta_\mu F_S=0$ for every $S$) to show that $\Tilde\Delta_\mu^2=0$ for every $\mu$. However, since $\Tilde\Delta_\mu$ is not of the form in Lemma~\ref{l:DeltaD} (unless, of course, $F_S=0$), it is not a BV Laplacian.
\end{Rem}

If $\Psi\colon\calM\to\calN$ is a diffeomorphism, we have the pushforward
\[
\Psi_*\colon\End(C^\infty(\calM))\to\End(C^\infty(\calN))
\]
defined by $\Psi_*P=\Psi_*\circ P\circ \Psi_*^{-1}$, where on the right-hand side $\Psi_*$ denotes the pushforward of functions.
The BV Laplacian transforms nicely under symplectomorphisms.
\begin{Prop}\label{p:psistarDeltaf}
For every symplectomorphism $\Psi$,
\[
\Psi_*\Delta_\mu=\Delta_{\Psi_*\mu}.
\]
\end{Prop}
\begin{proof}
The divergence operator changes under a diffeomorphism as
\[
\Psi_*\Div_\mu X_f= \Div_{\Psi_*\mu} \Psi_*X_f
\]
(see point~(\ref{p:div:Psi}) in Proposition~\ref{p:div}). If $\Psi$ is a symplectomorphism, we also have $\Psi_*X_f=X_{\Psi_*f}$.
\end{proof}
\begin{Rem}
In conjunction with Proposition~\ref{p:DeltaD}, this implies that under a symplectomorphism $\Psi\colon\calM\to\calM$, the leading term of the BV Laplacian on functions is invariant under symplectomorphisms. 
\end{Rem}
\begin{Rem}\label{r:symplsendcomp}
The proposition also implies that a symplectomorphism sends compatible densities to compatible densities.
\end{Rem}

\subsection{Digression: further properties of the BV Laplacian}
Even though we are not going to use this in the following, it is good to know that the properties stated in equations \eqref{e:trianglefg} and \eqref{e:triangleE} also hold for the global BV Laplacian.
\begin{Prop}\label{p:Deltaprop} For every even, nowhere vanishing density $\mu$, every functions $f$ and $g$, and every even function $S$, we have
\begin{subequations}
\begin{gather}
\Delta_\mu(fg)=(\Delta_\mu f)g +(-1)^f f\Delta_\mu g - (-1)^f (f,g),\label{e:Deltafg} \\
\Delta_\mu \EE^S = \left(\Delta_\mu S-\frac12(S,S)\right)\EE^S.\label{e:DeltaE}
\end{gather}
\end{subequations}
\end{Prop}
\begin{proof}
The identities may be proved by observing, thanks to \eqref{e:DeltaD}, that in every Darboux chart the BV Laplacian differs from $\triangle$ by a vector field.

They may also be proved directly using properties of the divergence operator (see point~(\ref{p:div:f}) in Proposition~\ref{p:div})
and of hamiltonian vector fields. To illustrate this, we prove the second identity. Since $X_{\EE^S}=\EE^SX_S$, we have
\[
\Div_\mu X_{\EE^S} = \Div_\mu(\EE^SX_S) = \EE^S \Div_\mu X_S + X_S(\EE^S),
\]
which proves the identity once we observe that $X_S(\EE^S)=-\EE^S(S,S)$.
\end{proof}

\begin{Rem}
As we have seen in Remark~\ref{r:tildeDelta}, it is possible to define an operator $\Tilde\Delta_\mu$ that squares to zero, no matter what $\mu$ is, so as to satisfy the extension of 
the property of equation \eqref{e:trianglesquare}. However, since this is obtained by adding to $\Delta_\mu$ a multiplication operator, we see that now \eqref{e:Deltafg} is no longer satisfied (unless, of course, $F_S=0$).
\end{Rem}

\section{The BV Laplacian on densities}\label{s:BVlapd}
A differential operator $P$ defined on functions on a (super)manifold $\calM$ can be extended to sections of a trivial line bundle $L$ over $\calM$ once an even, nowhere vanishing section $\lambda$ of $L$ has been chosen. For $\sigma\in\Gamma(L)$ there is a uniquely determined function $f$ such that $\sigma=f\lambda$, and we set
\[
P^{(\lambda)}\sigma\coloneqq P(f)\lambda.
\]
\begin{Lem}\label{l:leadL}
The leading term of the differential operator $P^{(\lambda)}$ does not depend on the choice of $\lambda$.
\end{Lem}
\begin{proof}
If $\Tilde\lambda$ is also a nowhere vanishing section, then there is a uniquely determined nowhere vanishing function $h$ with $\Tilde\lambda=h\lambda$. For $\sigma=f\Tilde\lambda=fh\lambda$, we have
\[
P^{(\Tilde\lambda)}\sigma = P(f)\Tilde\lambda = P(f)h\lambda.
\]
On the other hand, we have
\[
P^{(\lambda)}\sigma=P(fh)\lambda=((P_\text{leading}f)\,h +\cdots)\lambda,
\]
where the dots denote terms where less than the maximum number of derivatives hit $f$.
\end{proof}

We now want to extend the BV Laplacian $\Delta_\mu$ to $s$\ndash densities (see Appendix~\ref{a:dens} for a review).
Since we already have an even, nowhere vanishing section, $\mu$, of the density bundle, we can use it to get our reference section, $\mu^s$, of the $s$\ndash density bundle.\footnote{Since we are interested also in nonintegral $s$, in particular $s=\frac12$, it is essential to require that $\calM$ is oriented.} That is, 
\begin{Def}
For an $s$\ndash density $\sigma$, which we can uniquely write as $\sigma=f\mu^s$, we set\footnote{ We avoid the more pedantic notation $\Delta_\mu^{(\mu^s)}$ on the left-hand side.%The fact that $\Delta_\mu$ has two different meanings on the two sides of the equation should cause no confusion.
}
\[
\Delta_\mu^{(s)}\sigma\coloneqq \Delta_\mu(f)\mu^s.
\]
%If we want to stress that $\Delta_\mu$ acts on $s$\ndash densities we write $\Delta_\mu^{(s)}$.
\end{Def}
Note that $\Delta_\mu^{(0)}$ is the same as $\Delta_\mu$.

\begin{Rem}
An immediate consequence of this definition is that $\Delta_\mu^{(s)}$ squares to zero on $s$\ndash densities if and only if $\Delta_\mu$ does so on functions (i.e., $\mu$ is compatible with the odd symplectic structure).
\end{Rem}

Proposition~\ref{p:DeltaD} and Lemma~\ref{l:leadL} immediately imply
\begin{Prop}\label{p:Deltashmu} 
The BV Laplacian on $s$\ndash densities is a second-order differential operator, and its leading term is independent of the choice of reference density.
\end{Prop}

Proposition~\ref{p:psistarDeltaf} immediately implies
\begin{Prop}\label{p:psistarDeltas} 
For every symplectomorphism $\Psi$ and for every $s$,
\[
\Psi_*\Delta^{(s)}_\mu=\Delta^{(s)}_{\Psi_*\mu}.
\]
\end{Prop}
\begin{Rem}
In particular, this implies that, for every $s$, under a symplectomorphism $\Psi\colon\calM\to\calM$, the leading term of the BV Laplacian on $s$\ndash densities is invariant under symplectomorphisms. 
\end{Rem}

Note that the product of an $s$\ndash density $\sigma$ and a $(1-s)$\ndash density $\tau$ yields a density, which can be integrated (if it has compact support). Given a differential operator $P$ on $s$\ndash densities, we define its transpose $P^t$ as the differential operator on $(1-s)$\ndash densities that satisfies
\[
\int_{\calM} P\sigma\,\tau=(-1)^{P\,\sigma}\int_{\calM} \sigma\,P^t\tau
\] 
for all $s$\ndash densities $\sigma$ with compact support (the upper index ``$P\,\sigma$'' denotes the product of the degrees; see footnote~\fullref{f:signs}). The case $s=\frac12$ is special, since transposition becomes an endomorphism on the space of differential operators. 
\begin{Prop}\label{p:Deltatranspose}
We have
\[
(\Delta_\mu^{(s)})^t=\Delta_\mu^{(1-s)}
\]
In particular, $\Delta_\mu^{(\frac12)}$ is symmetric.
\end{Prop}
\begin{proof}
Write $\sigma=f\mu^s$ and $\tau=g\mu^{1-s}$. Then
\[
\Delta_\mu^{(s)}\sigma\,\tau = \Delta_\mu f\, g\,\mu.
\]
But
\begin{gather*}
2\Delta_\mu f\, g=\Div_\mu X_f\, g = (-1)^{(f+1)g} g\, \Div_\mu X_f\\
= (-1)^{(f+1)g} (\Div_\mu(gX_f) - (-1)^{(f+1)g}X_f(g))\\
= (-1)^{(f+1)g} \Div_\mu(gX_f) + (-1)^f (f,g),
\end{gather*}
where we have used part (\ref{p:div:f}) of Proposition~\ref{p:div}. Similarly,
\[
\sigma\, \Delta_\mu^{(1-s)}\tau = f\,\Delta_\mu g\,\mu,
\]
and
\begin{gather*}
2f\,\Delta_\mu g = f\,\Div_\mu X_g = \Div_\mu(f X_g)-(-1)^{f(g+1)}X_g(f)\\
=  \Div_\mu(f X_g) + (-1)^{fg+f+g} (g,f)
=  \Div_\mu(f X_g) + (f,g).
\end{gather*}
Therefore, 
\[
2\Delta_\mu f\, g-(-1)^f2f\,\Delta_\mu g= (-1)^{(f+1)g} \Div_\mu(gX_f)-(-1)^f \Div_\mu(f X_g),
\]
which implies, thanks to part (\ref{p:div:int}) of Proposition~\ref{p:div}, that\footnote{Note the different sign in \eqref{e:Deltafg}, which instead yields
\[
\int_\calM(\Delta_\mu^{(s)}\sigma\,\tau+(-1)^\sigma \sigma\, \Delta_\mu^{(1-s)}\tau)
=\int_\calM(\Delta_\mu f\, g+(-1)^f f\,\Delta_\mu g)\,\mu=(-1)^f\int_\calM (f,g)\,\mu.
\]
Subtracting the two identities yields the interesting formula
\[
\int_\calM (f,g)\,\mu = 2 \int_\calM f\,\Delta_\mu g\,\mu,
\]
which in particular shows that the integral of the BV bracket of two functions (at least one of which has compact support) against the density $\mu$ vanishes if one of the two functions is in the kernel of $\Delta_\mu$.}
\[
\int_\calM(\Delta_\mu^{(s)}\sigma\,\tau-(-1)^\sigma \sigma\, \Delta_\mu^{(1-s)}\tau)=0.
\]
\end{proof}
\begin{Rem}
By inspection in the proof one sees that it is essential that we have used the same density $\mu$ to define the BV Laplacian on functions and to extend it to $s$\ndash densities.
\end{Rem}

We now come to the fundamental consequence of this proposition.
\begin{Lem}
In each Darboux chart, there is an odd function $G_\mu$, depending on $\mu$, such that
\[
\Delta_\mu^{(\frac12)}\sigma = (\triangle f + G_\mu f)\,\mu_\text{stand}^{\frac12},
\]
where we have written $\sigma=f\,\mu_\text{stand}^{\frac12}$.
\end{Lem}
\begin{proof}
Since $\Delta_\mu^{(\frac12)}$ is a second-order differential operator and its leading term is independent of $\mu$, we have
\[
\Delta_\mu^{(\frac12)}\sigma = (\triangle f + Y_\mu(f) + G_\mu\, f)\,\mu_\text{stand}^{\frac12},
\]
for some vector field $Y_\mu$ and some function $G_\mu$. Since $\triangle$ has constant coefficients, it is symmetric. Moroever, since $Y_\mu$ is a vector field, we have $Y_\mu^t=-Y_\mu+$some function. Since $\Delta_\mu^{(\frac12)}$ is symmetric, we get $Y_\mu=0$.
\end{proof}
This implies the following
\begin{Cor}\label{c:Deltacan}
For every compatible $\mu$, the BV Laplacian on half-densities has a canonical form in every Darboux chart:
\[
\Delta_\mu^{(\frac12)}\sigma = \triangle f \,\mu_\text{stand}^{\frac12},
\]
where we have written $\sigma=f\,\mu_\text{stand}^{\frac12}$. In particular, $\Delta_\mu^{(\frac12)}$ is independent of the choice of the compatible density $\mu$.
\end{Cor}
\begin{proof}
Using $\triangle^2=0$, $G_\mu^2=0$, and \eqref{e:trianglefg}, we get
\[
0=(\Delta_\mu^{(\frac12)})^2\sigma = (\triangle G_\mu\,f + (G_\mu,f))\mu_\text{stand}^{\frac12}.
\]
Therefore, for every function $f$, we have
\[
\triangle G_\mu\,f + (G_\mu,f)=0.
\]
In particular, for $f=1$, we get $\triangle G_\mu=0$, so the condition simplifies to $(G_\mu,f)=0$ for every $f$. This implies that $G_\mu$ is constant, but then it must vanish because it is odd.\footnote{This argument fails if one wants to work in families, as in this case odd constants are allowed.}
\end{proof}

%\begin{Rem}
In every Darboux chart, we can consider the BV Laplacian $\Delta_{\mu_\text{stand}}^{(\frac12)}$, which sends $
f\,\mu_\text{stand}^{\frac12}$ to $\triangle f \,\mu_\text{stand}^{\frac12}$. If we move from this Darboux chart to another one via the transition map $\phi$, the operator goes to $\phi_*\Delta_{\mu_\text{stand}}^{(\frac12)}$, which, by Proposition~\ref{p:psistarDeltas}, is $\Delta_{\phi_*\mu_\text{stand}}^{(\frac12)}$. However, 
%since $\mu_\text{stand}$ is a compatible density, by the corollary 
since $\mu_\text{stand}$ is a compatible density (and therefore, by Remark~\ref{r:symplsendcomp}, also $\phi_* \mu_\text{stand}$ is a compatible density), by the corollary
this is again
$\Delta_{\mu_\text{stand}}^{(\frac12)}$ (in the new chart). Therefore, we have the
\begin{ThmDef}\label{td:Deltacan}
There is a \textsf{canonical BV operator} $\Delta$ acting on half-densities, defined as $\Delta_{\mu_\text{stand}}^{(\frac12)}$ in every Darboux chart.
\end{ThmDef}
Note that, by its very definition, the canonical BV operator is invariant under symplectomorphisms and satisfies $\Delta^2=0$.

%\end{Rem}
Corollary~\ref{c:Deltacan} may now be reformulated as 
\begin{Thm}\label{t:compdeltadelta}
For every compatible $\mu$, the BV Laplacian on half-densities is equal to the canonical BV operator:
\[
\Delta_\mu^{(\frac12)}=\Delta.
\]
\end{Thm}
By this theorem we have that the canonical BV operator is a BV Laplacian on half\ndash densities whenever the odd symplectic manifold admits compatible densities,\footnote{Unlike in the even symplectic case, where a canonical choice of density is provided by the top exterior power of the symplectic form, there is no canonical density in the odd case.} which we have not proved yet.

On the other hand, what is remarkable is that $\Delta$ is canonically defined (by Theorem/Definition~\ref{td:Deltacan}) without choosing any density (in particular it would exist even if there were no compatible densities). 
%We now have a different way of characte

%\begin{Rem}
The canonical BV operator gives another way of defining the BV Laplacian on functions. Namely, given an even, nowhere vanishing density $\mu$, we define a differential operator $\Hat\Delta_\mu$ on functions via
\begin{equation}\label{e:hatDelta}
\Hat\Delta_\mu f= \Delta(f\mu^{\frac12})\,\mu^{-\frac12}.
\end{equation}
Note that $\Hat\Delta_\mu^2=0$ for every $\mu$. However, it is not a BV Laplacian in general. We actually have the 
%\end{Rem}
\begin{Lem}\label{l:hattilde}
For every $\mu$, %we have
\[
\Hat\Delta_\mu = \Tilde\Delta_\mu
\]
in every Darboux chart, where $\Tilde\Delta_\mu$ is the operator defined in Remark~\ref{r:tildeDelta}.
\end{Lem}
\begin{proof}
Writing $\mu=\EE^S\mu_\text{stand}$, we have
\[
\Hat\Delta_\mu f\,\EE^{\frac S2}\mu_\text{stand}^{\frac12}= \Delta(f\EE^{\frac S2}\mu_\text{stand}^{\frac12})
=\triangle(f\EE^{\frac S2})\,\mu_\text{stand}^{\frac12}.
\]
One can easily compute
\[
\triangle(f\EE^{\frac S2})=
\left(\triangle f - \frac12(S,f)+\frac12 F_S\right)
\EE^{\frac S2}
=\Tilde\Delta_\mu f\,\EE^{\frac S2},
\]
 with $F_S$ as in \eqref{e:FS}.
\end{proof}
This yields the following
\begin{Cor}\label{c:mucomp}
Let $\mu$ be an even, nowhere vanishing density. Then
\begin{enumerate}
\item\label{c:mucomp:tilde} $\Tilde\Delta_\mu$ is globally defined;
\item\label{c:mucomp:hat} $\mu$ is compatible if and only if $\Hat\Delta_\mu=\Delta_\mu$;
\item\label{c:mucomp:square} $\mu$ is compatible if and only if $\Delta\mu^{\frac12}=0$.
\end{enumerate}
\end{Cor}
\begin{proof}
The first two properties are now trivial. We prove the third. 
%If $\mu$ is compatible, by (\ref{c:mucomp:hat}), equation
%\eqref{e:hatDelta} becomes
%\[
%\Delta_\mu f\,\mu^{\frac12}= \Delta(f\mu^{\frac12}).
%\]
%Inserting $f=1$ yields $\Delta\mu^{\frac12}=0$. On the other hand, it he last condition is verified, 
Setting $f=1$ in \eqref{e:hatDelta} yields $\Delta(\mu^{\frac12})=\Hat\Delta_\mu 1\, \mu^\frac12$. Lemma~\ref{l:hattilde} then implies that $\Delta(\mu^{\frac12})=0$ if and only if in every Darboux chart $F_S=\Tilde\Delta_\mu 1 = 0$.
By Remark~\ref{r:DeltaD} this happens if and only if $\mu$ is compatible.
\end{proof}

\subsection{Digression: Darboux expression of $\Delta_\mu^{(s)}$}
In a Darboux chart we have $\mu=\EE^S\mu_\text{stand}$. If $\sigma=f\,\mu_\text{stand}^s$, we then have
$\sigma= \EE^{-sS}f\,\mu^s$, so
\[
\Delta_\mu^{(s)}\sigma = \Delta_\mu( \EE^{-sS}f)\,\mu^s.
\]
By using \eqref{e:DeltaD} and the properties of $\triangle$ in \eqref{es:triangle},
we get
\[
\Delta_\mu^{(s)}\sigma = \left(
\triangle f +\left(s-\frac12\right)(S,f)-\left(
s\triangle S +\frac12 s(s-1)(S,S)
\right)f
\right)\mu_\text{stand}^s.
\]
One can now explictly check that $(\Delta_\mu^{(s)})^t=\Delta_\mu^{(1-s)}$. Moreover, one explicitly sees that the first-order term vanishes for every $S$ if and only if $s=\frac12$. Finally, for $\mu$ compatible, i.e., $\Delta S = \frac14(S,S)$, the zeroth-order term becomes $\frac14(2s-1)s(S,S)$, which vanishes for $s=\frac12$ or $s=0$.

This formula, for $s=\frac12$, yields a different way to prove Corollary~\ref{c:Deltacan} and therefore Theorem/Definition~\ref{td:Deltacan} and Theorem~\ref{t:compdeltadelta}. However, proving this formula is a bit more laborious than proving Proposition~\ref{p:Deltatranspose}.

% digression: integration on Lagr subm

\section{Odd cotangent bundles}\label{s:oct}
We now focus on an odd cotangent bundle $\Pi T^*M$,\footnote{Here $\Pi$ denotes the change-of-parity functor. It simply means that the fiber coordinates are now regarded as odd.} where $M$ is an ordinary manifold. We will see that it is easy to show the existence of compatible densities in this case. Since every odd symplectic manifold is symplectomorphic to the odd cotangent bundle of its body \cite{Sch93}, this implies that every odd symplectic manifold has a compatible density and therefore that the canonical BV operator is indeed a BV Laplacian on half-densities.  

The first observation is that functions on $\Pi T^*M$ are the same as multivector fields on $M$: $C^\infty(\Pi T^*M)=\Gamma(\Lambda^\bullet TM)$ as super algebras.\footnote{\label{f-mv}The interested reader might also appreciate that $(C^\infty(\Pi T^*M),(\ ,\ ))$ and $(\Gamma(\Lambda^\bullet TM),[\ ,\ ]_\text{SN})$ are naturally isomorphic odd Poisson algebras, where $(\ ,\ )$ is the canonical BV bracket, associated to the canonical symplectic form, and $[\ ,\ ]_\text{SN}$ is the Schouten--Nijenhuis bracket.}

The second observation is about the density bundle on $\Pi T^*M$.
If $\phi_{\alpha\beta}$ denotes a transition function on $M$, then the corresponding transition function on $T^*M$ has a fiber part given by $((\dd\phi_{\alpha\beta})^*)^{-1}$. Therefore, the Berezinian of the corresponding transition function for $\Pi T^*M$ is $(\dd\phi_{\alpha\beta})^2$. This means that densities on $\Pi T^*M$ transform like $2$\ndash densities on $M$. Since we assume $M$ to be oriented, we also have that half\ndash densities on $\Pi T^*M$ transform like top forms on $M$. More precisely, we have a canonical isomorphism
\[
\Gamma(\Dens^\frac12(\Pi T^*M))\simeq\Gamma(\Lambda^\bullet TM)\otimes_{C^\infty(M)} \Omega^n(M),
\]
where $n$ is the dimension of $M$. The next observation is that we have a canonical isomorphism (of $C^\infty(M)$\ndash modules)
\[
\phi\colon\begin{array}[t]{ccc}
\Gamma(\Lambda^\bullet TM)\otimes_{C^\infty(M)} \Omega^n(M) &\to & \Omega^{n-\bullet}(M)\\
X\otimes v &\mapsto & \iota_X v
\end{array}
\]
We can use this isomorphism to define an operator on half-densities:
\[
\calD \coloneqq \phi^{-1}\circ\dd\circ\phi,
\]
where $\dd$ is the de~Rham differential.
One immediately sees that $\calD^2=0$ and that $\calD v=0$ for every top form $v$ on $M$. % (regarded as a half\ndash density on $\Pi T^* M$).

\begin{Prop}\label{p:cancanBV}
The canonical operator $\calD$ is the same as the canonical BV operator $\Delta$ on half-densities. % on $\Pi T^*M$.
\end{Prop}
\begin{proof}
We just have to check how $\calD$ acts in Darboux coordinates. Let $q^1,\dots,q^n$ denote coordinates on a chart of $M$ and let $p_1,\dots,p_n$ be the corresponding odd fiber coordinates. By linearity it is enough to consider a half\ndash density $\sigma=g\,\mu_\text{stand}$ with $g(p,q)$ of the form $f(q)\,p_{i_1}\cdots\, p_{i_k}$, with pairwise distinct indices $i_j$. In the identification described above this corresponds to $\sigma = f(q)\,\de_{i_1}\wedge\dots\wedge\de_{i_k}\otimes\dd^nq$, so
\[
\phi(\sigma) = f(q)\,\iota_{\de_{i_1}}\cdots\,\iota_{\de_{i_k}}\dd^nq.
\]
We then have
\[
\dd\phi(\sigma) = \dd f(q)\,\iota_{\de_{i_1}}\cdots\,\iota_{\de_{i_k}}\dd^nq=
\de_if(q)\,\dd q^i\wedge\iota_{\de_{i_1}}\cdots\,\iota_{\de_{i_k}}\dd^nq,
\]
since the form $\iota_{\de_{i_1}}\cdots\,\iota_{\de_{i_k}}\dd^nq$ has a constant coefficient, so it is closed. If the index $i$ is different from all the $i_j$s, the corresponding term vanishes, since $\dd q^i$ is already contained in $\dd^nq$. Otherwise, we have
\[
\begin{split}
\dd q^{i_j}\wedge\iota_{\de_{i_1}}\cdots\,\iota_{\de_{i_k}}\dd^nq
&= (-1)^{j-1} \dd q^{i_j} \wedge\iota_{\de_{i_j}} \iota_{\de_{i_1}}\cdots\,\Hat\iota_{\de_{i_j}} 
\cdots\,\iota_{\de_{i_k}}\dd^nq\\
&=
(-1)^{j-1} \iota_{\de_{i_1}}\cdots\,\Hat\iota_{\de_{i_j}} 
\cdots\,\iota_{\de_{i_k}}\dd^nq,
\end{split}
\]
where the caret denotes omission. %Since $\dd q^{i_j} \wedge\iota_{\de_{i_j}} = \id -\iota_{\de_{i_j}} \dd q^{i_j} \wedge$
On the other hand, 
\[
\de_{i_j}(p_{i_1}\cdots\, p_{i_k})=(-1)^{j-1} p_{i_1}\cdots\, \Hat p_{i_j}\cdots\, p_{i_k}.
\]
Therefore, 
\[
\phi(\triangle g\,\mu_\text{stand})=\dd\phi(g\,\mu_\text{stand}),
\]
which shows that $\Delta=\calD$.
\end{proof}
We then have $\Delta v=0 $ for every top form $v$ on $M$ (regarded as a half\ndash density on $\Pi T^* M$).
As a notation, we write $\mu_v$ for the even, nowhere vanishing density $v^2$ on $\Pi T^*M$ associated to a volume form $v$ on $M$. We then have
\[
\Delta\mu_v^{\frac12}=0,
\]
which shows, by point (\ref{c:mucomp:square}) of Corollary~\ref{c:mucomp},
 that $\mu_v$ is a compatible density. We then have the
 \begin{Prop}\label{p:cotcom}
 An odd cotangent bundle possesses compatible densities.
 \end{Prop}
 This implies, by Theorem~\ref{t:compdeltadelta}, that $\Delta=\Delta^{(\frac12)}_{\mu_v}$.
Given a volume form $v$, we can now use \eqref{e:hatDelta} to get $\Delta_{\mu_v}$ on functions. It is an easy exercise to see that
\[
\Delta_{\mu_v} = \phi_v^{-1}\circ\dd\circ\phi_v,
\]
where $\phi_v$ is the ($v$\ndash dependent!) isomorphism
\[
\phi_v\colon\begin{array}[t]{ccc}
C^\infty(\Pi T^*M)=
\Gamma(\Lambda^\bullet TM) &\to & \Omega^{n-\bullet}(M)\\
X &\mapsto & \iota_X v
\end{array}
\]
Note that, for a vector field $X$, we have $\Delta_{\mu_v}X=\Div_vX$. For this reason, $\Delta_{\mu_v}$ is interpreted as the extension to multivector fields of the divergence operator with respect to $v$.

\section{Conclusion of the argument}
It is a theorem (see \cite{Sch93})
that every odd symplectic manifold $\calM$ is, noncanonically, isomorphic to the odd cotangent bundle $\Pi T^*M$ of its body $M$.\footnote{This is a global Darboux theorem. If one is familiar with the proof of Darboux's theorem via Moser's trick, the reason why this now works globally is that it turns out that one has to integrate a vector field in the odd directions, which has no problem of definition domain.}

On $\Pi T^*M$ we have compatible densities (Proposition~\ref{p:cotcom}), e.g., $\mu_v$ for a volume form $v$ on $M$.
If $\Psi\colon\Pi T^*M\to\calM$ is a symplectomorphism, which exists by the above mentioned theorem, then $\Psi_*\mu_v$ is a compatible density on $\calM$ (Remark~\ref{r:symplsendcomp}).

We therefore have---also using Theorem~\ref{t:compdeltadelta}, equation \eqref{e:hatDelta}, and point (\ref{c:mucomp:hat}) of Corollary~\ref{c:mucomp}---our final result:
\begin{Thm}
Every odd symplectic manifold possesses compatible densities. For each compatible density $\mu$, we have
\[
\Delta_\mu^{(\frac12)}=\Delta,
\]
with $\Delta$ the canonical BV operator on half\ndash densities, and the BV Laplacian $\Delta_\mu$ on functions may be recovered from
\[
\Delta_\mu f\,\mu^{\frac12}= \Delta(f\mu^{\frac12}).
\]
\end{Thm}

% digression application

\appendix

\section{Some background details}\label{a:somebackgroundmaterial}
\subsection{The Lie derivative}
If $\Xi$ is an object (e.g., function, vector field, density) for which a notion of pushforward $\Psi_*$ under a diffeomorphism $\Psi$ exists, we define the Lie derivative with respect to a vector field $X$ as
\[
\LL_X \Xi\coloneqq \left.\fde t\right\rvert_{t=0} (\Phi^X_{-t})_*\Xi.
\]
For an even vector field $X$, the diffeomorphism $\Phi^X_{t}$ is the flow of $X$ at time $t$. If $X$ is odd, the variable $t$ is also odd (so the evaluation at zero in the formula is redundant), and we define the morphism $\Phi^X\colon \Pi\bbR\times\calM\to\calM$ via\footnote{This defines the flow of $X$ at odd time $t$ if and only if $[X,X]=0$. Otherwise, the flow property $\Phi^X_{t+s}=\Phi^X_t\circ\Phi^X_{s}$ is not satisfied, so in general $\Phi^X$ is not a flow.} 
\[
(\Phi^X)^*\colon\begin{array}[t]{ccc}
C^\infty(\calM) & \to & C^\infty(\Pi\bbR\times\calM)\\
f & \mapsto & f + tX(f)
\end{array}
\]
where $t$ is the odd coordinate on $\Pi\bbR$.\footnote{Here $\Pi$ denotes the change-of-parity functor: $\Pi\bbR$ is the superdomain with one odd co\-ordinate.}
Note that in local coordinates $z^\mu$, with $X=X^\mu{\small\fde{z^\mu}}$, we have in both cases
\begin{equation}\label{e:PhiXexp}
(\Phi^X)^\mu \coloneqq (\Phi^X)^*z^\mu = z^\mu + t X^\mu + O(t^2),
\end{equation}
where of course $O(t^2)=0$ in the odd case.
Also note that in both cases we have
\[
\Psi_*(\Phi^X)_* = (\Phi^{\Psi_*X})_*\Psi_*
\]
for every diffeomorphism $\Psi$.\footnote{In the even case, one just observes how solutions of ODEs are mapped under diffeomorphisms. In the odd case, it is just an immediate application of the chain rule.} This implies that
\[
\LL_{\Psi_*X} \Psi_*\Xi = \Psi_* \LL_X\Xi.
\]
If two objects $\Xi_1$ and $\Xi_2$ as above can be multiplied, we immediately get
\begin{equation}\label{e:Lxixi}
\LL_X(\Xi_1\Xi_2)=\LL_X \Xi_1\,\Xi_2 + (-1)^{X\Xi_1} \Xi_1\, \LL_X \Xi_2.
\end{equation}
On a function $f$ or on a vector field $Y$, the Lie derivative is readily computed as
\[
\LL_Xf=X(f)\quad\text{and}\quad\LL_XY = [X,Y].
\]

\subsection{Densities}\label{a:dens}
A density is an object that transforms in such a way that its integration may be defined. For this we have to recall that, if $\Psi$ is a diffeomorphism between two superdomains, then the change-of-variables formula for berezinian integration involves the Berezinian $\Ber(\Psi)$ of the jacobian matrix of $\Psi$ times the sign of the jacobian determinant of the reduction to the body of $\Psi$. Since we assume throughout that the body of our supermanifold is oriented, we will not have this factor. Therefore, given an oriented atlas for $\calM$ that has transition functions $\phi_{\alpha\beta}$, we define the density bundle $\Dens\calM$ as the line bundle with transition functions $\Ber(\phi_{\alpha\beta})^{-1}$. More generally, for any real number $s$, we define the line bundle $\Dens^s\calM$ of $s$\ndash densities
with transition functions $\Ber(\phi_{\alpha\beta})^{-s}$. 
Note that $1$\ndash densities are the same as densities, and $0$\ndash densities are the same as functions.

With this definition the integral of a compactly supported\footnote{The support is defined in terms of the coefficient body functions. Namely, an object (function, vector field, density,\dots) is compactly supported when in its expansion in odd variables all the coefficients are compactly supported functions on the body.} density $\mu$ over $\calM$ is defined as usual (by choosing a partition of unity subordinated to the atlas and by berezinian integration in each chart).

The pushforward of $s$\ndash densities under diffeomorphisms is also defined. If we go to charts, the pushforward is given by the pushforward of the representing function times the Berezinian of the transformation to the power $-s$.

%We also get the following useful

Armed with the pushforward, we can define the Lie derivative of an $s$\ndash density. In particular, we need the following
\begin{Lem}
In local coordinates $p_1,\dots,p_m$ and 
$q^1,\dots,q^n$ (odd and even, respectively), the Lie derivative of the standard density $\mu_\text{stand}=\dd^nq\,\dd^mp$ is
\[
\LL_X {\mu_\text{stand}}= (\de_i X^i-(-1)^X\de^iX_i)\, {\mu_\text{stand}},
\]
where we used the expansion $X= X^i\de_i+X_i\de^i$.
\end{Lem}
\begin{proof}
Using \eqref{e:PhiXexp}, we see that\footnote{By $\delta^i_j$ we mean the Kronecker delta, which is equal to $1$ for $i=j$ and to $0$ otherwise.}
\[
\begin{split}
\Ber(\Phi^X_{-t})&=
\Ber\begin{pmatrix}
\delta_i^j - \de_i(tX^j) & 0\\
0 & \delta^i_j -\de^i(tX_j)
\end{pmatrix}
+O(t^2)\\
&=
- \de_i(tX^i) -\de^i(tX_i) +O(t^2)\\
&=
- t\de_i(X^i) -(-1)^Xt\de^i(X_i) +O(t^2),
\end{split}
\]
using that the paritiy of $t$ is the same as that of $X$.
Applying the definitions of Lie derivative and of pushforward of a density yields the result.
\end{proof}
Applying \eqref{e:Lxixi} to the case of densities, we get
\begin{equation}\label{e:Lhmu}
\LL_X(h\mu) = X(h)\,\mu+(-1)^{Xh} h\,\LL_X\mu
\end{equation}
for every function $h$. In particular, we get the general formula for the Lie derivative in local coordinates:
\[
\begin{split}
\LL_X \mu&= (X(h)+(-1)^{Xh} h (\de_i X^i-(-1)^X\de^iX_i))\, {\mu_\text{stand}}\\
&= (X(h)+ (\de_i X^i-(-1)^X\de^iX_i)h)\, {\mu_\text{stand}}\\
&=X(h)\mu_\text{stand} + (\de_i X^i-(-1)^X\de^iX_i)\mu,
\end{split}
\]
with $\mu=h\mu_\text{stand}$. This implies the following
\begin{Lem}
For every function $f$, vector field $X$, and density $\mu$, we have
\[
\LL_{fX}\mu = f\LL_X\mu+(-1)^{fX}X(f)\mu.
\]
\end{Lem}
\begin{proof}
It is enough to check the formula in local coordinates. We have
\[
\begin{split}
\LL_{fX} \mu&= fX(h)\mu_\text{stand} + (\de_i (fX^i)-(-1)^{f+X}\de^i(fX_i))\mu\\
&= fX(h)\mu_\text{stand} + f\,(\de_i X^i-(-1)^X\de^iX_i)\mu +  (\de_i fX^i-(-1)^{f+X}\de^ifX_i)\mu\\
&= f\LL_X\mu + ((-1)^{Xf} X^i\de_if-(-1)^{f+X}(-1)^{(f+1)(X+1)}X_i\de^if)\mu\\
&= f\LL_X\mu+(-1)^{fX}X(f)\mu.
\end{split}
\]
\end{proof}

In the case of densities, the pushforward is related to the change-of-variables formula. Namely, if $\Psi\colon\calM\to\calN$ is a diffeomorphism (preserving the orientation of the bodies), then
\begin{equation}\label{e:cov}
\int_\calN\Psi_*\mu=\int_\calM\mu
\end{equation}
for every compactly supported density $\mu$.
This leads to %the following
\begin{Lem}
For every vector field $X$ and density $\mu$, one of which is compactly supported, we have
\[
\int_\calM\LL_X\mu=0.
\]
\end{Lem}
\begin{proof}
First consider the case when $\mu$ is compactly supported. In this case, $(\Phi^X_{-t})_*\mu$ and hence $\LL_X\mu$ 
are also compactly supported. We then have
\[
\int_\calM\LL_X\mu=\int_\calM \left.\fde t\right\rvert_{t=0} (\Phi^X_{-t})_*\mu
= \left.\fde t\right\rvert_{t=0} \int_\calM (\Phi^X_{-t})_*\mu
= \left.\fde t\right\rvert_{t=0} \int_\calM \mu=0,
\]
where we have also used \eqref{e:cov}.

If on the other hand $\mu$ is not compactly supported, but $X$ is, we replace $\mu$ with $\Tilde\mu=\rho\mu$, where $\rho$ is a compactly supported bump function which is identically equal to $1$ on the support of $X$. Since $\LL_X\mu=\LL_X\Tilde\mu$, we get the result by the previous case.
\end{proof}

If we have an even, nowhere vanishing density $\mu$, then every other density can be written as $f\mu$ for a uniquely determined function $f$. Therefore, we can define the divergence operator via the following formula:
\[
\Div_\mu X \,\mu= \LL_X\mu.
\]
The properties of the Lie derivative immediately imply properties for the divergence operator:
\begin{Prop}\label{p:div}
Let $\mu$ be an even, nowhere vanishing density $\mu$ and $X$ a vector field. Then
\begin{enumerate}
\item\label{p:div:Psi} for every diffeomorphism $\Psi$,
\[
\Psi_*\Div_\mu X= \Div_{\Psi_*\mu} \Psi_*X;
\]
\item\label{p:div:hmu} for every even, nowhere vanishing function $h$,
\[
\Div_{h\mu} X = \Div_\mu X + \frac1hX(h);
\]
\item\label{p:div:f} for every function $f$,
\[
\Div_\mu (fX) = f\Div_\mu X+(-1)^{fX} X(f);
\]
\item\label{p:div:int} under the additional assumption that $X$ is compactly supported,
\[
\int_\calM \Div_\mu X \,\mu=0.
\]
\end{enumerate} 
\end{Prop}
Moreover, from the local coordinate expression for the Lie derivative, we get
\begin{subequations}
\begin{align}
\Div_{\mu_\text{stand}} X &= \de_i X^i-(-1)^X\de^iX_i,\label{se:divstand}\\
\intertext{and, rearranging the terms,}
\Div_{\mu} X &= \frac1h(\de_i (hX^i)-(-1)^X\de^i(hX_i)),
\end{align}
\end{subequations}
for $\mu=h\mu_\text{stand}$.

\section{Applications of the BV formalism}\label{a:app}
We recall here the main reason why the BV formalism was introduced: to study the invariance of certain integrals.

We start considering the case of Section~\ref{s-standLap}, where we have Darboux coordinates $p_1,\dots,p_n$ and  $q^1,\dots,q^n$. For the sake of the argument, we rearrange each pair $(p_i,q^i)$ into a new pair $(x_i,y^i)$, 
with $x_i=p_i$ and $y^i=q^i$ for some $i$s and $x_i=q^i$ and $y^i=p_i$ for the other $i$s.
Note that %the symplectic form now reads $\omega=\dd x_i\dd y^i$ and that
$x_i$ has opposite parity to $y^i$ (but we do not insist on $y^i$ being even). We now use the shorthand notation 
\[
\de_i = \fde{y^i}\quad\text{and}\quad\de^i=\fde{x_i}.
\]
Note that the BV Laplacian still reads $\triangle=\de_i\de^i$.

If $\psi$ is an odd function of the variables $y$, it makes sense, in terms of parity, to set $x_i=(-1)^{x_i}\de_i\psi$ (we will explain the reason for the signs in a moment).
For a function $f$ in the $(x,y)$ variables,
we define
\[
\int_{\calL_\psi} f \coloneqq \int f|_{x_j=(-1)^{x_j}\de_j\psi} \; \dd^n y,
\]
where $|_{x_j=(-1)^{x_j}\de_j\psi}$ means that we set each $x$ variable to the corresponding right-hand-side expression.
The reason for the notation is that $\calL_\psi$ is a Lagrangian submanifold determined by $\psi$. In fact,
$\omega=\dd x_i\dd y^i=\dd (x_i\dd y^i)$ and 
\[
(x_i\dd y^i)|_{x_i=(-1)^{x_i}\de_i\psi}=\dd\psi.
\]
The first result is the BV lemma.
\begin{Lem}
Suppose that $f=\triangle g$ and $g$ is integrable. Then $\int_{\calL_\psi} f=0$ for every $\psi$.
\end{Lem}
\begin{proof}
We have
\[
\begin{split}
\sum_i \de_i(\de^i g)|_{x_j=(-1)^{x_j}\de_j\psi}
& =
\sum_i (\de_i\de^i g)|_{x_j=(-1)^{x_j}\de_j\psi}\\
& \phantom=\  +
\sum_{ik} (\de_i( (-1)^{x_k}\de_k\psi)\,  \de^k\de^i g)|_{x_j=(-1)^{x_j}\de_j\psi}.
\end{split}
\]
Since $\sum_{ik}  (-1)^{x_k}(\de_i\de_k\psi\,  \de^k\de^i g)=0$, as one can easily see by exchanging the indices, we get
\[
 (\triangle g)|_{x_j=(-1)^{x_j}\de_j\psi}
=
\sum_i \de_i(\de^i g)|_{x_j=(-1)^{x_j}\de_j\psi},
\]
so $\int_{\calL_\psi} f $ vanishes.
\end{proof}
This leads to the fundamental BV theorem.
\begin{Thm}\label{t:fundBV}
Let $\psi_t$ be a family of odd functions in the $y$ variables depending smoothly on the even parameter $t$.
If $f$ is integrable, on every $\calL_{\psi_t}$, and $\triangle f=0$, then
\[
I_t \coloneqq \int_{\calL_{\psi_t}} f 
\]
is constant.
\end{Thm}
\begin{proof}
We have
\[
\dot I_t =  \int \sum_i 
((-1)^{x_i}\de_i\dot\psi\,  \de^if)|_{x_j=(-1)^{x_j}\de_j\psi} \; \dd^n y,
\]
where the dot denotes derivative with respect to $t$. {}From
\[
\triangle(\dot\psi\, f) = \sum_i\de_i\de^i(\dot\psi\, f)= \sum_i(-1)^{x_i} \de_i(\dot\psi\, \de^if)
= \sum_i (-1)^{x_i}\de_i\dot\psi\,  \de^if -\dot\psi\,\triangle f,
\]
we get 
\[
\dot I_t = \int  \triangle(\dot\psi\, f)|_{x_j=(-1)^{x_j}\de_j\psi} \; \dd^n y=
\int_{\calL_{\psi_t}} \triangle(\dot\psi\, f) = 0,
\]
where we have used the previous lemma.
\end{proof}

The BV theorem is used in quantum gauge theories,\footnote{In this application, the odd symplectic manifold is actually infinite dimensional, so $\triangle$ is not defined. What one does is either to proceed formally or to regularize the theory, e.g., by replacing the space of fields with a finite\ndash dimensional approximation.} where the choice of $\psi$ (called the gauge-fixing fermion) corresponds to the choice of a gauge fixing and the invariance under deformations of this choice yields the independence of the theory from the gauge fixing, assuming of course that $\triangle f=0$.

Using the results of Sections~\ref{s:BVlapf} and~\ref{s:BVlapd}, one can extend the BV lemma and the BV theorem globally.
The main observation is that the restriction of a half\ndash density on an odd symplectic manifold to a Lagrangian submanifold defines a density there \cite{Sch93}.
\begin{Thm}\label{t:globBV}
Let $\calM$ be an odd symplectic manifold and $\sigma$ a half\ndash density. The following hold:
\begin{enumerate}
\item If $\sigma=\Delta\tau$ for some $\tau$, then
\[
\int_\calL\sigma=0
\]
for every a Lagrangian submanifold $\calL$ of $\calM$ on which $\sigma$ is integrable.
\item If $\Delta\sigma =0$ and $\calL_t$ is a smooth family of Lagrangian submanifolds of $\calM$ on which $\sigma$ is integrable, then
\[
I_t \coloneqq \int_{\calL_t} \sigma
\]
is constant.
\end{enumerate}
\end{Thm}

In the case of an odd cotangent bundle $\Pi T^*M$, one can show that every Lagrangian submanifold is a smooth deformation of an odd conormal bundle $\Pi N^* C$, where $C$ is a submanifold of $M$. Moreover, one can show that 
\[
\int_{\Pi N^* C}\sigma = \int_C \phi(\sigma)
\]
for every half\ndash density $\sigma$. By Stokes' theorem and the characterization of the canonical BV operator as in Proposition~\ref{p:cancanBV} (i.e., $\phi\circ\Delta=\dd\circ\phi$), one gets that 
\[
\int_{\Pi N^* C_1}\sigma=\int_{\Pi N^* C_2}\sigma
\]
if $C_1$ and $C_2$ are homologous and $\Delta\sigma=0$.

By Schwarz' theorem one can then generalize the global BV theorem a bit.
\begin{Thm}
Let $\calM$ be an odd symplectic manifold and $\sigma$ a half\ndash density satisfying $\Delta\sigma =0$. Then
\[
\int_{\calL_1} \sigma=\int_{\calL_2} \sigma
\]
whenever $\calL_2$ can be obtained from $\calL_1$ by a combination of smooth deformations and homologous changes of the body, assuming that $\sigma$ is integrable on each intermediate step.
\end{Thm}

%{s:oct}

\section{Historical remarks}\label{a:historem}
The BV formalism was developed by Batalin and Vilkovisky \cite{BV77,BV81} as a generalization of the BRST formalism \cite{BRS,T}, which in turn put gauge fixing and the Faddeev--Popov determinant \cite{FP67} into a cohomological setting. They constructud the BV Laplacian (as in our Section~\ref{s-standLap}) and proved their fundamental theorem (in our notes, Theorem~\ref{t:fundBV}). In addition, they showed that, under suitable assumptions, a physical action can be extended to a BV action $S$ satisfying the classical master equation $(S,S)=0$.

In \cite{Wit90} Witten recognized the relation between the BV operator on functions on an odd cotangent bundle and the divergence operator on multivector fields on the body (see the last paragraph of Section~\ref{s:oct}). The relation between the de~Rham differential and the divergence operator on multivector fields on an ordinary manifold was already known to Soviet mathematicians in the '80s (see, e.g., \cite{Kir80}). Also note that Bernshtein and Leites \cite{BL77} had already introduced in 1977 the notion of ``integral forms'' on supermanifolds---where this notion differs from that of differential forms---and defined the exterior differential for them as the divergence operator.

%in the supercontext: Bernstein-Leites in 1977 defined "integral forms" as multivector densities and defined the "exterior differential" for them simply as the (invariant in this case) divergence operator. 

Khudaverdian \cite{Khu91} was the first to give the global definition of the BV Laplacian on functions, see our equation \eqref{e:Deltamudiv}.

Khudaverdian was also the first to observe that an odd symplectic manifold always admits global Darboux coordinates. This was used by Schwarz in \cite{Sch93}, where he also observed how to globalize and extend the BV theorem (in our notes, Theorem~\ref{t:globBV}).

Finally, in \cite{Khu99,Khu04},\footnote{Note that O. M. Khudaverdian and H. M. Khudaverdian are just different spellings of the same name.} Khudaverdian showed the existence of a canonical BV Laplacian on half-densities (in our notes, this corresponds to Theorem/Definition~\ref{td:Deltacan}).

In a series of papers (among others, \cite{KV02a,KV02b,KV06}), Khudaverdian and Voronov further clarified the properties of the canonical BV Laplacian, essentially covering all the constructions we present in these notes (and more). In \cite{KV02a} they extended the construction of the odd Laplacian on half-densities to odd Poisson manifolds. In \cite{KV02b} they discussed the principal and subprincipal symbols of the BV Laplacians and their transposed operators (see our Section~\ref{s:BVlapd}). In \cite{KV06} they defined the canonical BV Laplacian on half-densities on an odd cotangent bundle in terms of the de~Rham differential of the corresponding differential forms and then showed that this operator is invariant under all symplectomorphisms---not just those coming from the base.

It is worth mentioning that \v Severa \cite{Sev06} 
% https://arxiv.org/abs/math/0506331
found a completely different construction for the canonical BV Laplacian. % in the $\bbZ$\ndash graded setting. 
The main observation is that the complex of differential forms on $\calM$ has two commuting coboundary operators: the de~Rham differential $\dd$ and the operator $\delta\coloneqq \omega\wedge\ $. It turns out that the associated spectral sequence lives up to the $E_2$\ndash term. More precisely, \v Severa shows that $E_1=H_\delta(\calM)$ is canonically isomorphic to $\Dens^{\frac12}(\calM)$ and that the induced coboundary operator $\dd_1$ vanishes, which implies that $E_2=E_1=\Dens^{\frac12}(\calM)$. Finally, 
\v Severa proves that the canonically induced coboundary operator $\dd_2$ is the canonical BV Laplacian and that all higher coboundary operators vanish. % (this is actually equivalent to $\Delta^2=0$).

\begin{comment}
It is worth mentioning that \v Severa \cite{Sev06} 
% https://arxiv.org/abs/math/0506331
found a completely different construction for the canonical BV Laplacian in the $\bbZ$\ndash graded setting. 
Namely, one assumes that local coordinates on $\calM$ have an additional $\bbZ$\ndash grading for which the odd symplectic form $\omega$ has degree $-1$. The complex of differential forms on $\calM$ is then bigraded and has two commuting coboundary operators of total degree $+1$: the de~Rham differential $\dd$ and the operator $\delta\coloneqq \omega\wedge\ $. It turns out that the associated spectral sequence lives up to the $E_2$\ndash term. More precisely, \v Severa shows that $E_1=H_\delta(\calM)$ is canonically isomorphic to $\Dens^{\frac12}(\calM)$ and that the induced coboundary operator $\dd_1$ vanishes, which implies that $E_2=E_1=\Dens^{\frac12}(\calM)$. Finally, 
\v Severa proves that the canonically induced coboundary operator $\dd_2$ is the canonical BV Laplacian and that all higher coboundary operators vanish. % (this is actually equivalent to $\Delta^2=0$).
\end{comment}

\end{document}